\documentclass[runningheads]{llncs}
\usepackage{graphicx}
\usepackage{amssymb}
\usepackage{amsmath}
\usepackage{hyperref}
\usepackage{xcolor}
\usepackage{enumerate}
\usepackage{listings}
\usepackage{complexity}
\usepackage{verbatim}

\DeclareMathOperator{\conv}{conv}

\begin{document}
\title{Shortest Paths on Polymatroids and Hypergraphic Polytopes}
%
%
\author{Jean Cardinal\inst{1}\orcidID{0000-0002-2312-0967} \and
Raphael Steiner\thanks{Research of R.S. supported by the SNSF Ambizione Grant No. 216071.}\inst{2}\orcidID{0000-0002-4234-6136}}
\authorrunning{J. Cardinal and R. Steiner}
%
\institute{Universit\'e libre de Bruxelles (ULB), Brussels, Belgium\\
  \email{jean.cardinal@ulb.be}\and
  ETH Zurich, Z\"{u}rich, Switzerland\\
\email{raphaelmario.steiner@inf.ethz.ch}}
\maketitle              
\sloppy
\begin{abstract}
  Base polytopes of polymatroids, also known as generalized permutohedra, are polytopes whose edges are parallel to a vector of the form $\mathbf{e}_i - \mathbf{e}_j$, where the $\{\mathbf{e}_i\}_{i\in [n]}$ are the canonical basis vectors of $\mathbb{R}^n$.
  
  We consider the following computational problem: Given two vertices of a generalized permutohedron $P$, find a shortest path between them on the skeleton of $P$, where the length of a path is its number of edges. This captures many known flip distance problems, such as computing the minimum number of exchanges between two spanning trees of a graph, the rotation distance between binary search trees, the flip distance between acyclic orientations of a graph, or rectangulations of a square. 
  
  We prove that this general problem is \NP-hard, even when restricted to very simple polymatroids in $\mathbb{R}^n$ defined by $O(n)$ inequalities. 
  Assuming $\P\not= \NP$, this rules out the existence of a computationally efficient simplex pivoting rule that would perform a minimum number of nondegenerate pivoting steps to an optimal solution of a linear program, even when the latter defines a polymatroid.
  
  We also prove that the shortest path problem is inapproximable when the polymatroid is specified via an evaluation oracle for a corresponding submodular function, which strengthens a recent result by Ito, Kakimura, Kamiyama, Kobayashi, Maezawa,  Nozaki, and Okamoto (ICALP'23).
  More precisely, we prove the \APX-hardness of the shortest path problem when the polymatroid is a hypergraphic polytope, whose vertices are in bijection with acyclic orientations of a given hypergraph. The shortest path problem then amounts to computing the flip distance between two acyclic orientations of a hypergraph. 

  On the positive side, we provide a polynomial-time approximation algorithm for the problem of  computing the flip distance between two acyclic orientations of a hypergraph, where the approximation factor is the maximum codegree of the hypergraph. Our result implies in particular an exact polynomial-time algorithm for the flip distance between two acyclic orientations of any linear hypergraph.
\end{abstract}

\section{Introduction}

We consider the problem of finding a shortest path between two vertices of a graph formed by the vertices and edges of an $n$-dimensional polytope, where the length of the path is its number of edges. Bounding the length of such paths has been the topic of intensive research since the birth of the theory of linear programming and the invention by Dantzig, around 1947, of the Simplex algorithm, that seeks solutions of a linear program by iteratively moving from one feasible solution to another, along the edges of the polytope. The performance of this method, parameterized by a so-called \emph{pivoting rule}, is still the subject of many open questions. In particular, it is not known whether there exists a pivoting rule that makes the simplex algorithm run in strongly polynomial time. In a series of recent papers, the question of whether any efficient algorithm could either find or approximate the shortest path to a solution has been studied, and several negative results have been proved~\cite{ACHKMS21,DKS22,IKK0MNO23,IKKKO22}. In particular, shortest paths between perfect matchings of a bipartite graph on the perfect matching polytope were shown to be hard to approximate~\cite{CS23}. This implies that no simplex pivoting rule can perform an approximately optimal number of (nondegenerate) pivoting steps, even on linear programs defined by totally unimodular matrices.

Another motivation for studying shortest paths on polytopes is the problem of computing the rotation distance between two binary search trees, or, equivalently, the flip distance between two triangulations of a convex polygon. The equivalence can be proved by a classical Catalan bijection that maps every triangulation to its dual tree. This problem, too, can be cast as finding a shortest path between two vertices of a polytope. The associated polytope is called the \emph{associahedron}, and has been the topic of numerous investigations in various mathematical contexts, including data structures, topology, and algebra~\cite{STT88,L04,P14,CSZ15}. The complexity status of the rotation distance problem is another widely open question. It is not known to be \NP-hard, but no polynomial-time algorithm is known either. The only known results are a simple polynomial-time 2-approximation algorithm and several polynomial-time fixed-parameter algorithms~\cite{CS09,LX23}.

These two problems, the existence of efficient pivoting rules and computing the rotation distance between binary trees, are our main motivations for introducing the question of computing the distance between vertices of 
\emph{base polytopes of polymatroids}, also known in the combinatorics literature as \emph{generalized permutohedra}. In what follows, we will refer to them simply as \emph{polymatroids}. From the optimization point of view, these are polytopes that generalize matroid base polytopes, while preserving the property that the vertices which are extremal with respect to a linear function can be found easily using a greedy algorithm~\cite{S03,F05}. Hence just as perfect matching polytopes, polymatroids yield another class of linear programs that are solvable in strongly polynomial time. While polymatroids have been introduced and studied by Edmonds as early as 1970~\cite{E03}, they play a key role in more recent developments in algebraic combinatorics, where they are defined as polytopes whose normal fan is a coarsening of the normal fan of the permutohedron~\cite{PRW08,P09}. They are fundamental objects, among others, in the study of combinatorial Hopf algebras~\cite{AA17,LR12,M19}. Associahedra are classical examples of polymatroids, hence the rotation distance between binary search trees is a special case of a shortest path problem on a polymatroid.

 Recently, Ito, Kakimura, Kamiyama, Kobayashi, Maezawa,  Nozaki, and Okamoto~\cite{IKK0MNO23} showed that unless $\P=\NP$, there is no polynomial-time algorithm for finding a shortest path between two vertices of a polymatroid, when it is given in the form of an oracle to its associated submodular rank function, a standard computation model for optimization on matroids and polymatroids. More precisely, they prove the hardness of computing the rotation distance between elimination trees on graphs, a problem that generalizes the rotation distance between binary search trees, and amounts to finding shortest paths on a class of polymatroids known as \emph{graph associahedra}~\cite{CD06,MP14}.
 
\subsection{Our results}

We strengthen the result by Ito et al.~\cite{IKK0MNO23} in two ways.

We first show that there exist simply-structured polymatroids in $\mathbb{R}^n$ described by a set of $2n+2$ inequalities, on which the shortest path problem is strongly \NP-hard. 
We show that the shortest path problem on those polytopes is equivalent to the known \emph{pure constant fixed charge transportation} problem, a special case of the classical fixed charge transportation problem.

We then consider the shortest path problem in the setting where the polymatroid is specified by an oracle to a submodular function. Strengthening the aforementioned hardness result by Ito et al., we show that in this computational model, the problem of finding a shortest path between two vertices of a polymatroid cannot even be efficiently approximated, unless $\P=\NP$.
To prove this result, we consider a special family of polymatroids defined by acyclic orientations of hypergraphs. These polytopes are known as \emph{hypergraphic polytopes}~\cite{BBM19,CHMMM23}. They form a natural generalization of polymatroids defined by acyclic orientations of graphs, and have recently been extensively studied both from an algebraic and combinatorial viewpoint~\cite{AA17,R22}. They also include the special cases of associahedra and graph associahedra mentioned above.

Finally, we provide a positive result about acyclic orientations of hypergraphs. We show that the shortest path problem on hypergraphic polytopes can be approximated within a factor equal to the maximum codegree of the hypergraph. This provides a new polynomial-time exact algorithm for the shortest path problem on hypergraphic polytopes corresponding to linear hypergraphs, i.e. hypergraphs in which any two hyperedges intersect in at most one vertex.

Section~\ref{sec:bg} sets the stage for the problems we tackle and details a number of related computational questions. In Section~\ref{sec:hard} we prove the hardness of the shortest path problem on a polymatroid defined by linear inequalities. Section~\ref{sec:ao} describes the necessary definitions and background related to acyclic orientations of hypergraphs and the corresponding flip graphs. In Sections~\ref{sec:hard2} and \ref{sec:approx}, we then present the proofs of the inapproximability result and the codegree approximation algorithm, respectively.

\subsection{Other related works}

Our contribution is also relevant in the context of \emph{combinatorial reconfiguration} problems~\cite{IDHPSUU11,H13,GIKO22}\footnote{See also \url{https://reconf.wikidot.com/}.}. In such problems, we are given two combinatorial objects together with local reconfiguration operations, and we are asked to decide whether it is possible to transform the first object into the other by iteratively applying a reconfiguration operation, and if it is, to find the minimum number of required steps. This problem is typically posed for structures defined as solutions of optimization problems, such as graph colorings or dominating sets~\cite{MN19}, independent sets~\cite{IKOSUY20}, cliques~\cite{IOO23}, bounded-degree spanning trees~\cite{BIKMOSW23}, $k$-edge-connected orientations~\cite{IIKKKMNOO23}, among others. Our results can be seen as contributions in this vein, although in the cases we consider, the question of whether any two objects are connected is trivial, since the the polytopal structure ensures connectivity (and in fact, $d$-connectivity, where $d$ is the dimension of the polytope, from Balinski's Theorem~\cite{B61}). 

\section{Definitions and background}
\label{sec:bg}

In what follows, we use the notation $\mathbf{p}=(p_1,p_2,\ldots ,p_n)$ for $n$-dimensional vectors, and let $[n]:=\{1,2,\ldots ,n\}$. 
We first consider the problem of finding a shortest path between two vertices of a polytope specified by linear inequalities.

\begin{problem}[Shortest paths on a polytope]
  \label{pb:sp}
 Given a polytope $P$ in $\mathbb{R}^n$, described by a system of linear inequalities, and two vertices $\mathbf{a}$ and $\mathbf{b}$ of $P$, each identified by $n$ tight inequalities, find a shortest path from $\mathbf{a}$ to $\mathbf{b}$ on the skeleton of $P$. Here the length of a path is defined as the number of edges.
\end{problem}

This problem is sometimes referred to as \emph{routing} on a polytope.
It is closely related to the complexity of the Simplex method and the Hirsch conjecture. In its polynomial version, the latter states that the diameter of $n$-polytopes with $m$ facets is bounded by a polynomial in $n$ and $m$. If the polynomial Hirsch conjecture holds, then Problem~\ref{pb:sp} lies in \NP.
It is known that Problem~\ref{pb:sp} is \NP-hard to approximate within any constant factor, even when $P$ is restricted to be the perfect matching polytope of a bipartite graph~\cite{ACHKMS21,CS23}. We are interested in proving similar results in the case where the polytope is a polymatroid.

\subsection{Polymatroids and generalized permutohedra}

\emph{Base polytopes of polymatroids} are also known as \emph{generalized permutohedra}. In what follows, with a slight abuse of terminology, we will simply refer to them as \emph{polymatroids}.
They can be characterized in different ways~\cite{S03,F05,DF10,ACEP20}, generalizing the characterizations of matroid base polytopes by their edge lengths and the matroid rank function~\cite{GGMS87,S03}.

Recall that a set function $f:2^{[n]}\to\mathbb{R}$ is called \emph{submodular} if the inequality $f(T)+f(U)\geq f(T\cap U)+f(T\cup U)$ holds for all pairs $U,T$ of subsets of $[n]$. We let $\mathbf{e}_i$ denote the $i$th canonical basis vector of $\mathbb{R}^n$.

\begin{theorem}[Characterizations of polymatroids]
  \label{thm:gp}
  Given a polytope $P\subset\mathbb{R}^n$, the following statements are equivalent:
  \begin{enumerate}
  \item Every edge of $P$ is parallel to $\mathbf{e}_j-\mathbf{e}_i$ for a pair $i,j\in [n]$.
  \item There exists a submodular function $f:2^{[n]}\to \mathbb{R}$ such that
    \[
    P = P_f := \left\{\mathbf{x}\in\mathbb{R}^n : \sum_{i\in U} x_i\leq f(U)\ \mbox{ for all } U\subset [n],\mbox{ and } \sum_{i\in [n]} x_i = f([n])\right\}.
    \]
    \end{enumerate}
\end{theorem}

The reader is referred to \cite{S03,F05,P09,AA17} for details on this classical statement. 
If we further ask the function $f$ to be nondecreasing and satisfy $f(\emptyset)=0$ and $f(U)\leq |U|$, it is called a \emph{rank function}, but we will not require these conditions here.
We focus on the case where $f$ only takes on integral values. In this case the vertices of $P$ have integer coordinates, hence we focus on \emph{integer polymatroids}.

\subsection{Shortest paths on polymatroids}
\label{sec:cases}

We consider a variant of Problem~\ref{pb:sp} in which the input polytope $P$ is encoded by a submodular function $f:2^{[n]}\to \mathbb{R}$. 

\begin{problem}[Shortest paths on a polymatroid]
  \label{pb:spgp}
Given an oracle to a submodular function $f$, and two vertices $\mathbf{a}$ and $\mathbf{b}$ of the polytope $P_f$, compute the length of a shortest path from $\mathbf{a}$ to $\mathbf{b}$ on the skeleton of $P_f$.
\end{problem}
Several interesting computational problems can be reduced to Problem~\ref{pb:spgp}.

\subsubsection{Matroids.}

Clearly, when $f$ is the rank function of a \emph{matroid} $M$ on ground-set $[n]$, then Problem~\ref{pb:spgp} consists of computing the base exchange distance between two bases $A$ and $B$ of $M$, which is simply $|A\setminus B|$ (or, symmetrically, $|B\setminus A|$). In fact, a shortest path between two vertices $\mathbf{x},\mathbf{y}$ of $P_f$ in this case can be found greedily using polynomially many oracle calls to the rank function $f$: We start by reading off the corresponding bases $A$ and $B$ as the supports of $\mathbf{x}$ and $\mathbf{y}$. We then greedily construct a sequence of bases $A=B_0,B_1,\ldots,B_t=B$ inducing a shortest path from $\mathbf{x}$ to $\mathbf{y}$ on $P_f$ as follows. For each $k \ge 0$ such that $B\setminus B_k\neq \emptyset$, we search through all potential element-pairs $(i,j) \in (B\setminus B_k)\times (B_k\setminus B)$ and test whether $(B_k\cup \{i\})\setminus \{j\}$ is a basis (by using the oracle to decide if $f((B_k\cup \{i\})\setminus \{j\})=f([n])$). Such a pair always exists by the basis exchange axiom for matroids. Once a suitable pair $(i,j)$ is found, we define $B_{k+1}:=(B_k\cup \{i\})\setminus \{j\}$ and proceed. Thus, Problem~\ref{pb:spgp} is efficiently solvable in the case of matroids. 

\subsubsection{Acyclic orientations of graphs and graphical zonotopes.}

Given a simple, connected graph $G=([n],E)$, consider the submodular function $f:2^{[n]} \rightarrow \mathbb{N}_0$, where for a vertex-subset $U \subseteq [n]$ the value of $f(U)$ is simply defined as the number of edges that have at least one endpoint in $U$. 
The polymatroid $P_f$ defined by this choice of $f$ is known as the \emph{graphical zonotope} of $G$~\cite{S73,G77,GZ83}. The vertices of $P_f$ are one-to-one with the acyclic orientations of $G$, and its edges are one-to-one with pairs of acyclic orientations that differ by a single edge flip. Given an acyclic orientation $A$ of $G$, the entries of its corresponding vertex of $P_f$ simply correspond to the in-degrees of the vertices in $A$.

In a way that is similar to the previous case, the flip distance problem for graphical zonotopes can be shown to be solvable easily: The flip distance between two acyclic orientations $A$ and $B$ of $G$ is exactly equal to the number of edges that are oriented differently in $A$ and $B$, and a sequence of $|A \Delta B|$ edge-reversals that transforms $A$ into $B$ while keeping all intermediate orientations acyclic can be efficiently computed (see e.g.~\cite{fukuda}). We claim that in fact, a shortest path between two given vertices $\mathbf{x},\mathbf{y}$ of $P_f$ can be found efficiently even when one is only given oracle access to the function $f$. To see this, note that given any pair $u,v \in [n]$, we can determine whether $u$ and $v$ are adjacent in $G$ by computing the value of $f(u)+f(v)-f(\{u,v\})$, which is $1$ if $u$ and $v$ are adjacent and $0$ otherwise. Thus, using polynomially oracle calls to $f$ we can determine $G$. Let $A$ and $B$ be the acyclic orientations of $G$ whose in-degree sequences correspond to $\mathbf{x}$ and $\mathbf{y}$. By repeatedly identifying a vertex of in-degree $0$ and removing it, we can efficiently reconstruct the orientations $A$ and $B$ of $G$ from $\mathbf{x}$ and $\mathbf{y}$. Finally, knowing $A, B$ and $G$ we can efficiently find a shortest flip-sequence of edge-reversals that transforms $A$ into $B$ (see~\cite{fukuda}) and map this to a shortest path from $\mathbf{x}$ to $\mathbf{y}$ on the skeleton of $P_f$. 

We refer the reader to Pilaud~\cite{P22}, and Cardinal, Hoang, Merino, and M\"utze~\cite{CHMMM23} for recent developments on the combinatorics of acyclic orientations of graphs. One of the contributions of the present paper is the analysis of a hypergraph generalization of this problem.

\subsubsection{Triangulations and associahedra.} 

When the submodular function is defined as 
\begin{equation}
\label{f:assoc}
    f(U) = |\{ (i,j)\in [n]^2 : i<j \mbox{ and } U\cap [i,j]\not=\emptyset \}|,
\end{equation}
then the polymatroids are Loday's associahedra~\cite{STT88,P14,L04,CSZ15,PSZ23}.
Problem~\ref{pb:spgp} in this special case is the problem of computing the flip distance between two triangulations of a convex polygon, where a flip consists of replacing the diagonal of a quadrilateral formed by two adjacent triangles by its other diagonal. Equivalently, this amounts to computing the rotation distance between the dual binary trees.
The more general problems of computing flip distances between triangulations of arbitrary point sets or polygons are known to be \NP-hard~\cite{LP15,AMP15}. \FPT\ algorithms have been designed~\cite{KSX17,CS09}.

\subsubsection{Elimination trees and graph associahedra.} 

An \emph{elimination tree} of a graph $G$ is obtained by selecting any vertex $v$ of $G$ as the root, and attaching elimination trees obtained by recursing on each of the connected component of $G-v$. A natural rotation operation connects some pairs of elimination trees, and the rotation graph on the elimination trees is the skeleton of the so-called \emph{graph associahedron} of $G$.

The graph associahedron of $G$ can be specified by a submodular function $f$ defined as follows. Let $V$ be the vertex set of $G$, and define $\mathcal{B}(G)\subset 2^V$ as the collection of nonempty subsets $S\subseteq V$ such that $G[S]$ is connected. The collection $\mathcal{B}(G)$ is known as the \emph{graphical building set} of $G$. The submodular function defining the associahedron of $G$ is then defined as
\begin{equation}
\label{f:grassoc}
    f(U) = |\{ S\in\mathcal{B}(G) : U\cap S\not=\emptyset \}|.
\end{equation}

Graph associahedra were defined by Carr and Devadoss~\cite{CD06} and Postnikov~\cite{P09}, and also considered by Postnikov, Rainer, and Williams~\cite{PRW08}.
The structure of their skeleton was the subject of several recent works~\cite{MP14,CMM22,CPV22}. The flip distance problem in this case is known to be \NP-hard from the recent result of Ito et al.~\cite{IKK0MNO23}.

\subsubsection{Partial permutations and stellohedra.} 

Stellohedra are associahedra of stars~\cite{MP14,CP16}, and their vertices, the elimination trees of a star, can be shown to be one-to-one with \emph{partial permutations} of the set of leaves~\cite{CMM22}. The rotations can be interpreted in that case as either adjacent transpositions or adding or removing an element in final position. The corresponding flip distance problem was recently shown to be solvable in polynomial time by Cardinal, Pournin, and Valencia-Pabon~\cite{CPV23}.

\subsubsection{Rectangulations and quotientopes.} 

A \emph{rectangulation} is a tesselation of the unit square by disjoint axis-aligned rectangles. Modulo natural equivalence classes, rectangulations can be shown to be in bijection with classes of pattern-avoiding permutations~\cite{MM23}. Similarly to triangulations, flip graphs on rectangulations can be defined and shown to be skeletons of polytopes~\cite{CSS18,M19}.
These polytopes have been studied in~\cite{LR12,R12} and belong to the wider family of \emph{quotientopes}~\cite{PS19,PPR21}. They are also polymatroids, and therefore the flip distance problems between rectangulations are also special cases of Problem~\ref{pb:spgp}. The complexity of these problems is open.

\medskip
Note that graphical zonotopes, associahedra, graph associahedra, and stellohedra are all special cases of \emph{hypergraphic polytopes}, that we consider extensively in Section~\ref{sec:ao}. 

\section{\NP-hardness of shortest paths on polymatroids defined by linear inequalities}
\label{sec:hard}

We prove that Problem~\ref{pb:sp} is \NP-hard even in the case where the polytope $P$ is a polymatroid in the sense of Theorem~\ref{thm:gp}. In fact, the polytope can be restricted to be a very special polymatroid.

\begin{theorem}
  \label{thm:main}
 Problem~\ref{pb:sp} is strongly \NP-hard even when the input polytope is a polymatroid defined as the intersection of an axis-aligned box with a hyperplane.
\end{theorem}

We now define precisely the family of polytopes we consider.
Let $\mathbf{a},\mathbf{b}\in\mathbb{N}^n$ such that $\sum_{i=1}^n a_i = \sum_{i=1}^n b_i = N$ for some $N\in\mathbb{N}$. In what follows, we let the vectors $\mathbf{l}$ and $\mathbf{u}$ be defined as $l_i := \min\{a_i,b_i\}$ and $u_i:=\max\{a_i,b_i\}$, respectively.
We consider the polytope
\[
P := \left\{\mathbf{x}\in\mathbb{R}^n : \mathbf{l} \leq \mathbf{x} \leq \mathbf{u}  \mbox{ and }  \sum_{i=1}^n x_i=N\right\} .
\]

\begin{problem}[Shortest path on a box]
  \label{pb:spbox}
  Given two vectors $\mathbf{a},\mathbf{b}\in\mathbb{N}^n$ such that $\sum_{i=1}^n a_i = \sum_{i=1}^n b_i = N$ for some $N\in\mathbb{N}$, and $P$ defined as above,
  find the shortest path between $\mathbf{a}$ and $\mathbf{b}$ on the skeleton of $P$.
  \end{problem}

\begin{lemma}
  \label{lem:hard}
  Problem~\ref{pb:spbox} is strongly \NP-hard.
\end{lemma}

We characterize the edges of $P$ exactly and label them with pairs $(i,j)\in [n]^2$.

\begin{lemma}
  \label{lem:flips}
  Two distinct vertices $\mathbf{x}$ and $\mathbf{y}$ of $P$ are adjacent on the skeleton of $P$ if and only if there exists two indices $i,j\in [n]$ such that:
  \begin{itemize}
  \item $y_k = x_k \in \{l_k,u_k\}$ for all $k \in [n]$ with $k\not= i,j$,
  \item $y_j = x_j + \delta_{ij}$,
  \item $y_i = x_i - \delta_{ij}$,
  \end{itemize}
  where $\delta_{ij} = \min\{u_j - x_j, x_i - l_i\}$.
  In that case, we say that the oriented edge $(\mathbf{x}, \mathbf{y})$ is an $(i,j)$-\emph{flip}.
\end{lemma}
\begin{proof}
  For the first direction of the stated equivalence, suppose that $\mathbf{x}$ and $\mathbf{y}$ are adjacent on the skeleton of $P$. Then the segment $\conv(\{\mathbf{x},\mathbf{y}\})$ forms an edge of~$P$ (and thus, a face of $P$ of dimension $1$). Since $P$ is the intersection of the axis-aligned box $Q:=[a_1,b_1] \times \cdots \times [a_n,b_n]$ with the hyperplane $H$ of equation $\sum_{i=1}^n x_i=N$, it follows that $\conv(\{\mathbf{x},\mathbf{y}\})$ is the intersection of a two-dimensional face of the box $Q$ and the hyperplane $H$.
  Such a two-dimensional face is an axis-aligned rectangle
  \[
    R = [l_i, u_i] \times [l_j, u_j]
  \]
  lying in a plane orthogonal to all $\mathbf{e}_k$ such that $k\not= i,j$.
  In this plane, the edge is simply the intersection of $R$ with a line 
  parallel to $\mathbf{e}_j-\mathbf{e}_i$.
  It is straightforward to check that the coordinates of its endpoints obey the equations above.

  Vice-versa, suppose that two distinct vertices $\mathbf{x},\mathbf{y}$ of $P$ satisfy the stated conditions, i.e., there exist $k, \ell \in [n]$ such that $y_k=x_k \in \{u_k,l_k\}$ for every $k \in [n]\setminus \{i,j\}$ and $y_j=x_j+\delta_{ij}, y_i=x_i+\delta_{ij}$, where $\delta_{ij}:=\min\{u_j-x_j,x_i-l_i\}$. We now claim that $\conv(\{\mathbf{x},\mathbf{y}\})$ is a face (and thus an edge) of $P$, as claimed. To see this, it suffices to note that the face $F \subseteq P$ defined as
  $$F=\left\{z \in \mathbb{R}^n\bigg\vert u_i \le z_i \le l_i, u_j \le z_j \le l_j,\forall k \in [n]\setminus\{i,j\}:z_k=x_k, \sum_{i=1}^{N}{z_i}=N\right\}$$ contains both $\mathbf{x}$ and $\mathbf{y}$ (and thus $\conv(\{\mathbf{x},\mathbf{y}\})$) and is $1$-dimensional. This shows that $\mathbf{x}, \mathbf{y}$ are adjacent on the skeleton of $P$, concluding the proof.
\qed\end{proof}

From Lemma~\ref{lem:flips}, the edges of $P$ are
parallel to $\mathbf{e}_j-\mathbf{e}_i$ for a pair $i,j\in [n]$.
Hence from Theorem~\ref{thm:gp}, $P$ is a polymatroid and
Problem~\ref{pb:spbox} can also be reduced to Problem~\ref{pb:spgp}.
The submodular function $f:2^{[n]}\to\mathbb{R}$ such that $P_f = P$ is
\[
f(U) = \min \left\{ N - \sum_{i\in [n]\setminus U} l_i, \sum_{i\in U} u_i \right\}.
\]

We say that an $(i,j)$-flip is a \emph{$(U,V)$-flip} for some subsets $U\subseteq [n]$ and $V\subseteq [n]$ if $i\in U$ and $j\in V$.
Let $S := \{i\in [n] : a_i = u_i\}$, and $T := \{i\in [n]: a_i=l_i\}=[n]\setminus S$. The set $S$ is therefore the set of indices of components of $\mathbf{a}$ that have to eventually decrease, while $T$ denote those components that have to eventually increase. 


We show that Problem~\ref{pb:spbox} can be cast as a special case of the \emph{fixed charge transportation} problem on the bipartite graph $S\times T$. This problem is a minimum cost maximum flow problem in which the cost of every edge is the sum of a variable cost proportional to the flow on the edge and a fixed charge for each edge with nonzero flow.
It has a long history dating back to the pioneering work from Hirsch and Dantzig~\cite{HD68}, and is known to be \NP-hard. The special case where the cost is simply the number of edges with nonzero flow, and there is no variable cost, is known as the \emph{pure constant fixed charge transportation} (PCFCT) problem.
The PCFCT problem has been considered in particular by Hultberg and Cardoso~\cite{HC97} under the name \emph{teacher assignment problem}, and by Schrenk, Finke, and Cung~\cite{SFC11}. In a sense, it is about finding the ''most degenerate'' solution to a set of flow and positivity constraints.

\begin{problem}[Pure constant fixed charge transportation (PCFCT)]
Given two sets $S$ and $T$ and integers $\{s_i\}_{i\in S}$ and $\{t_i\}_{i\in T}$ such that $\sum_{i\in S} s_i = \sum_{j\in T} t_j$, find reals $x_{ij}\geq 0$, for $(i,j)\in S\times T$ such that:
\begin{eqnarray*}
\sum_{j\in T} x_{ij} & = & s_i, \ \forall i\in S \\   
\sum_{i\in S} x_{ij} & = & t_j, \ \forall j\in T,  
\end{eqnarray*}
and the number of nonzero values $|\{(i,j)\in S\times T: x_{ij}>0\}|$ is minimum.
\end{problem}

The PCFCT problem is known to be strongly \NP-hard, see Hultberg and Cardoso~\cite{HC97}. 
Their proof is by reduction from the strongly \NP-complete 3-partition problem. 
They also observed that the problem can in fact be cast as a maximum partition problem. We define a \emph{same-sum partition} of $S$ and $T$ as a collection of pairs $(S_{\ell},T_{\ell})$ of non-empty subsets of $S$ and $T$, such that $\{S_{\ell}\}$ is a partition of $S$, $\{T_{\ell}\}$ is a partition of $T$, and for each $\ell$, we have $\sum_{i\in S_{\ell}} s_i = \sum_{j\in T_{\ell}} t_j$.

\begin{lemma}[Hultberg and Cardoso~\cite{HC97}]
\label{lem:partition}
    An optimal solution of the PCFCT problem has value $|S|+|T|-m$, where $m$ is the maximum size of a same-sum partition of $S$ and $T$.
\end{lemma}

The proof of Lemma~\ref{lem:partition} is simple. We give an example in Figure~\ref{fig:partition}.

\begin{figure}
    \centering
    \includegraphics[page=2, scale=.8]{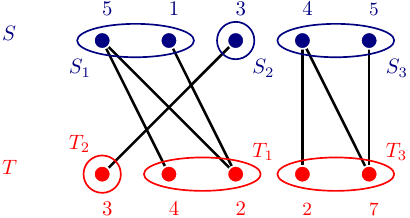}
    \caption{Illustration of Lemma~\ref{lem:partition}. The top and bottom rows show the values of the $s_i$ and $t_j$, respectively. There exists a same-sum partition of size $m=3$, and a corresponding solution of value $|S|+|T|-m=5+5-3=7$.}
    \label{fig:partition}
\end{figure}

We now prove the following key lemma, establishing that the PCFCT problem can be reduced to Problem~\ref{pb:spbox}.

\begin{lemma}
  \label{lem:sp2fct}
  Any instance of the PCFCT problem can be converted into an equivalent instance $(\mathbf{a},\mathbf{b})$ of Problem~\ref{pb:spbox}, so that the length of a shortest path is exactly equal to the value of an optimal solution of the PCFCT instance.
\end{lemma}
\begin{proof}
  Given an instance $(\{s_i\}_{i\in S}, \{t_j\}_{j\in T}\}$ of the PCFCT problem, we construct an equivalent instance $(\mathbf{a},\mathbf{b})$ of Problem~\ref{pb:spbox} as follows.
  First, we can assume without loss of generality that $S=[n_1]$ and 
  $T=\{n_1+1, n_1+2,\ldots ,n_1+n_2\}$, and let $n=n_1+n_2$.   
  We now let $\mathbf{a}=(s_1,s_2,\ldots ,s_{n_1},0,0,\ldots ,0)$, and
  $\mathbf{b}=(0,0,\ldots ,0,t_1,t_2,\ldots ,t_{n_2})$, and $u_i=s_i$ for each $i\in S$, $u_j=t_j$ for each $j\in T$, and finally $l_i=0$ for all $i\in S\cup T$.

  We first show that given an optimal solution to the PCFCT problem, we can find a path between $\mathbf{a}$ and $\mathbf{b}$ whose length is the value of the solution.
  
  Recall from Lemma~\ref{lem:partition} that an optimal solution can be mapped to a same-sum partition $(S_{\ell},T_{\ell})_{\ell=1}^{m}$ of $S$ and $T$ of maximum size, such that the optimal value of the PCFCT equals $|S|+|T|-m$. To construct a path on the skeleton of $P$ from $\mathbf{a}$ to $\mathbf{b}$ of length $|S|+|T|-m$, we consider the following simple algorithm. 
  
  The algorithm keeps track of two subsets $X$ and $Y$ of $S$ and $T$, respectively, the current vertex $\mathbf{x}$ of the polytope $P$, as well as a special ``active'' element $v$. Initially, the algorithm sets $X\leftarrow S, Y \leftarrow T$, $\mathbf{x}\leftarrow \mathbf{a}$ and picks $v \in S_1$ arbitrarily. 

  During the process, we maintain the following invariants.

  \begin{enumerate}
      \item $\mathbf{x}$ is a vertex of $P$.
    \item If $\ell$ is the index such that $v \in S_\ell \cup T_\ell$, then $X\subseteq S_\ell \cup \dots \cup S_m$, $Y\subseteq T_\ell \cup \dots \cup T_m$ and $v \in X\cap S_\ell$ or $v \in Y \cap S_\ell$.
      \item $X=\{k \in S|x_k>0\}$ and $Y=\{k \in T|x_k<t_k\}$.
      \item $x_k=u_k=s_k$ for every $k \in X\setminus \{v\}$ and $x_k=l_k=0$ for every $k \in Y\setminus \{v\}$.
      \item For every $\ell=1,\ldots,m$, it holds that $$\sum_{k \in X \cap S_\ell}{x_k}=\sum_{k \in Y \cap S_\ell}{(t_k-x_k)}.$$
  \end{enumerate}

  Note that these invariants hold after the initialization of $(\mathbf{x},v,X,y)$. Now, for $\ell=1,\ldots,m$, while $v \in S_\ell \cup T_\ell$, the algorithm updates $(\mathbf{x},v,X,y)$ as follows.

 It starts by picking a pair $(i,j) \in (X\cap S_\ell)\times (Y\cap T_\ell)$ such that $v\in \{i,j\}$. Note that this is always possible, since invariants (2), (3) and (5) together enforce that $X \cap S_\ell$ and $Y \cap S_\ell$ are non-empty.
 
The algorithm now performs an $(i,j)$-flip at the vertex $\mathbf{x}$ of $P$ (see the definition from Lemma~\ref{lem:flips}). Concretely, we put $\delta \leftarrow \min\{x_i, t_j-x_j\}$ and update $x_i\leftarrow x_i-\delta, x_j \leftarrow x_j+\delta$. Since $x_k \in \{l_k,u_k\}$ for every $k \notin \{i,j\}$ by the invariants (3) and (4), Lemma~\ref{lem:flips} shows that this change of $\mathbf{x}$ indeed corresponds to moving along an edge of the polytope $P$, in particular invariant (1) is preserved. Note that by definition of $\delta$, we have $x_i=0$ or $x_j=l_j=t_j$ after the flip. The algorithm now proceeds according to the following cases:

\begin{itemize}
    \item If $x_i=0$ and $x_j<t_j$ after the flip, then the algorithm updates $X\leftarrow X\setminus \{i\}$ and $v\leftarrow j$. 
    \item Symmetrically, if $x_i>0$ and $x_j=t_j$ after the flip, then the algorithm updates $Y\leftarrow Y\setminus \{j\}$ and $v\leftarrow i$.
    \item If $x_i=0$ and $x_j=t_j$ after the flip, then we claim that we must have $X\cap S_\ell=\{i\}$  and $Y\cap S_\ell=\{j\}$. Indeed, if $x_i=0$ and $x_j=t_j$ after the flip, then we must have had $x_i=t_j-x_j$ before the flip. Using invariant (5), this implies
    $$\sum_{k \in (X\cap S_\ell)\setminus \{i\}}{x_k}=\sum_{k \in (Y \cap S_\ell)\setminus \{j\}}{(t_k-x_k)}.$$However, by invariant (4) we have $x_k=s_k$ for every $k \in (X\cap S_\ell)\setminus \{i\}$ and $x_k=0$ for every $k \in (Y\cap S_\ell)\setminus \{j\}$. This implies that refining the same-sum partition of $(S,T)$ by replacing the pair $(S_\ell,T_\ell)$ with the two pairs $((X\cap S_\ell)\setminus \{i\},(X\cap S_\ell)\setminus \{j\})$ and $((S_\ell\setminus X)\cup \{i\},(T_\ell\setminus Y)\cup \{j\})$ would yield a larger same-sum partition, a contradiction --- unless $(X\cap S_\ell)\setminus \{i\}$ and $(Y\cap T_\ell)\setminus \{j\}$ are empty. This proves that $X\cap S_\ell=\{i\}$  and $Y\cap S_\ell=\{j\}$, as claimed above. 

    Finally, the algorithm proceeds to update $X\leftarrow X\setminus \{i\}, Y\leftarrow Y\setminus \{j\}$. If $\ell<m$, then it updates $v$ to be any vertex of $S_{\ell+1}$, and the for loop now continues with the new index $\ell+1$ instead of $\ell$. If $\ell=m$, then the algorithm stops and returns the sequence of all vertices of $P$ visited during the process. 
\end{itemize}

Pause to note that each of the three types of updates above preserve our invariants (1)--(5). It is clear from the above update rules that in every step of the algorithm, whenever we take a step along an edge of $P$ to a new vertex, the value of $|X|+|Y|$ decreases by one, and in fact sometimes this value decreases by two, namely exactly once in each of the $m$ iterations of the for loop, when we depart the while-loop for a fixed $\ell=1,\ldots,m$ and make the intersection of $X$ and $Y$ with $S_\ell$ and $T_\ell$ empty.

  The above implies that the algorithm finishes after a total of $|S|+|T|-m$ steps, and returns a sequence of $|S|+|T|-m+1$ vertices of $P$, any two consecutive of which are adjacent on the skeleton. Note that at each point of the algorithm, by invariant (3) it holds that $|X|+|Y|$ is the hamming-distance between the current vertex $\mathbf{x}$ of $P$ and the target-vertex $\mathbf{b}$. In particular, the last vertex of this sequence is equal to $\mathbf{b}$. It follows that the sequence of vertices of $P$ that we output forms a path from $\mathbf{a}$ to $\mathbf{b}$ on the skeleton of $P$ that has length $|S|+|T|-m$, as desired. This concludes the first part of the proof of the lemma.

\medskip
  
  We now turn to the second part of the proof, and show that given a path between $\mathbf{a}$ and $\mathbf{b}$ on the skeleton of $P$, we can find a solution of the PCFCT instance whose value is at most the length of the path. 
  Let us consider all flips on this path, and construct the undirected graph $G$ on $S\cup T$ in which two vertices $i,j$ are adjacent whenever there is an $(i,j)$-flip in the path. Note that some flips might not be $(S,T)$-flips, so $G$ is not necessarily bipartite.
  Each connected component $S_{\ell}\cup T_{\ell}$ of $G$ must be such that $\sum_{i\in S_{\ell}} s_i = \sum_{j\in T_{\ell}} t_j$. Hence the connected components induce a same-sum partition of $S$ and $T$, of some size $m$. So from Lemma~\ref{lem:partition}, we can construct a solution of the PCFCT problem of value $|S|+|T|-m=n-m$. On the other hand, the number of flips in each connected component $S_{\ell}\cup T_{\ell}$ of $G$ is at least $|S_{\ell}|+|T_{\ell}|-1$, hence the total number of flips on the path is at least $\sum_{\ell =1}^m (|S_{\ell}|+|T_{\ell}|-1) = n-m$.\qed
  \end{proof}

Combining Lemma~\ref{lem:sp2fct} with the fact that the PCFCT problem is strongly \NP-hard proves Lemma~\ref{lem:hard}, and therefore Theorem~\ref{thm:main}.

\section{Flip distances between acyclic orientations of hypergraphs}
\label{sec:ao}

We introduce hypergraphic polytopes, the associated flip distance problem, and our two results on the complexity of this problem: an \APX-hardness proof, and an approximation algorithm.

\subsection{Hypergraphic polytopes}

Every hypergraph $H=(V, \mathcal{E})$, where $\mathcal{E}\subseteq 2^V\setminus\{\emptyset\}$, induces a corresponding submodular function $f_H$ according to

\begin{equation}
\label{f:hypergraph}
f_H(U) := |\{ e\in\mathcal{E} : U\cap e\not=\emptyset \}|
\end{equation}

for every $U \subseteq V$. The base polytope induced by the polymatroid associated with $f_H$ is then called the \emph{hypergraphic polytope} of the hypergraph $H$, and denoted by $P_H$. Hypergraphic polytopes and the corresponding polymatroids have been studied in various guises in the past 50 years as covering hypermatroids~\cite{H74}, Boolean polymatroids~\cite{VW93}, or relational submodular functions~\cite{AA17}. Hypergraphic polytopes form an important family of generalized permutohedra~\cite{PRW08,P09,AA17,R22}. 

Interestingly, this definition encompasses most of the special cases we described in Section~\ref{sec:cases}. When $H$ is a graph, the polymatroids are the graphical zonotopes, whose vertices are in bijection with acyclic orientations of the graph. When $H$ is defined as 
\[
H = \{\{i,i+1,\ldots ,j-1,j\} : 1\leq i<j\leq n\},
\]
then one can check that Eqn.~\ref{f:hypergraph} matches Eqn.~\ref{f:assoc}, hence the hypergraphic polytope is nothing but the associahedron. When the hyperedges correspond with building sets of some graph $G$, then the definition matches that of Eqn.~\ref{f:grassoc}, and the polymatroid is a graph associahedron.
We refer to~\cite{PRW08,P09,BBM19,R22,CHMMM23} for detailed treatments of the connections of the polytopes $P_H$ to combinatorial structures. 

\subsection{Acyclic orientations of hypergraphs}

It is known that the vertices of the hypergraphic polytope $P_H$ are in one-to-one correspondence with so-called \emph{acyclic orientations} of the hypergraph $H$~\cite{BBM19,R22,CHMMM23}, which we define now.

An \emph{orientation} of a hypergraph $H=(V,\mathcal{E})$ is a mapping $h:\mathcal{E}\rightarrow V$ such that $h(e)\in e$ for every $e \in \mathcal{E}$. The element $h(e)$ is also called the \emph{head} of an edge $e$. With every orientation $h$ of $H$ we can associate a directed graph $D_h$ that has vertex set $h$ and an arc from a vertex $u$ to a vertex $v$ if and only if there exists some $e \in \mathcal{E}$ such that $h(e)=v$ and $u \in e\setminus \{v\}$. Intuitively, $D_h$ is the union, over all $e \in \mathcal{E}$, of the inwards-oriented stars centered at the head of $e$ and whose leaves are all other elements of $e$.
With this notion at hand, we call an orientation $h$ of $H$ \emph{acyclic} if $D_h$ is an acyclic digraph.

\begin{figure}
    \centering
    \includegraphics[page=4,scale=.6]{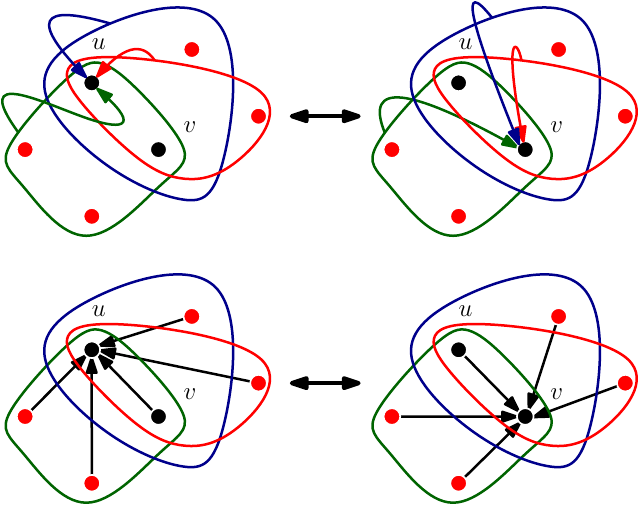}
    \caption{An example of flip, showing all hyperedges $e\in\mathcal{E}$ such that $e \supseteq \{u,v\}$ and satisfying $h_1(e)=u$, and arrows pointing to their heads $h_1(e)$ (left) and $h_2(e)$ (right). The bottom figure shows the corresponding parts of the digraphs $D_{h_1}$ and $D_{h_2}$.}
    \label{fig:flip}
\end{figure}

Adjacency of vertices on the skeleton of $P_H$ has also been characterized, see Theorem~\ref{thm:adjacency} below. In order to state it, we need one more definition, which is that of a \emph{flip}. Given two distinct acyclic orientations $h_1$, $h_2$ of a hypergraph $H=(V,\mathcal{E})$, we say that $h_2$ is obtained from $h_1$ by a \emph{flip} at an ordered pair $(u, v) \in V^2$ of distinct vertices, if $h_2$ is obtained from $h_1$ by redefining the head of every hyperedge $e\in \mathcal{E}$ with $e \supseteq \{u,v\}$ that satisfies $h_1(e)=u$ as $h_2(e)=v$, while keeping all other head assignments the same as in $h_1$. We note that the above definition of a flip is symmetric, in the following sense: If $h_2$ can be obtained from $h_1$ by flipping at $(u,v)$, then $h_1$ is obtained from $h_2$ by flipping at $(v,u)$. 
Also, note that 
when $h_2$ is obtained from $h_1$ by flipping at $(u,v)$, then a flip at $(v,u)$ for $h_1$ is impossible. Thus, in the following we can say without ambiguity that two acyclic orientations $h_1$ and $h_2$ are obtained from each other by flipping the unordered pair $\{u,v\}$ of distinct vertices if $h_2$ can be obtained from $h_1$ by flipping $(u,v)$ or $(v,u)$. An illustration is given in Figure~\ref{fig:flip}

\noindent The following statement can be deduced as a special case of Theorem~2.18 in~\cite{BBM19}.

\begin{theorem}
\label{thm:adjacency}
    Let $H=(V,\mathcal{E})$ be a hypergraph. Two vertices $x, y$ of the polytope $P_H$ corresponding to acyclic orientations $h_1, h_2$ of $H$ are adjacent on the skeleton of $P_H$ if and only if there exists a pair $\{u,v\} \subseteq V$ of distinct vertices such that $h_1$ and $h_2$ can be obtained from each other by flipping $\{u,v\}$. 
\end{theorem}

\subsection{Flip distances between acyclic orientations of hypergraphs}

Theorem~\ref{thm:adjacency} implies that the computation of shortest paths on the skeleton of a hypergraphic polytope $P_H$ boils down to a flip distance problem between acyclic orientations of $H$, as follows:

\begin{problem}\label{prob:hyperdistance}
    Given a hypergraph $H=(V,\mathcal{E})$ and two acyclic orientations $h_1, h_2$ of $H$ as input, find a shortest sequence vertex-pairs $\{u_i,v_i\}$, $i=1,\ldots ,\ell$ such that $h_2$ can be obtained from $h_1$ by flipping, in order, $\{u_1,v_1\},\ldots,\{u_\ell,v_\ell\}$. 
\end{problem}

Problem~\ref{prob:hyperdistance} generalizes several of the previous problems mentioned in Section~\ref{sec:bg}. It clearly includes the special case of flip distances between acyclic orientations of graphs, when the hypergraph $H$ is a graph. The flip distance between triangulations or binary search trees is the special case of Problem~\ref{prob:hyperdistance} in which the underlying hypergraph is the set of intervals in $[n]$, as mentioned above. 

In the following, we will prove several results regarding the computational complexity of Problem~\ref{prob:hyperdistance}.
On the negative side, we show that approximating the length of a shortest flip sequence between two acyclic orientations of a general hypergraph $H$ is \APX-hard, even restricted to hypergraphs $H$ of bounded degree and bounded hyperedge size.

\begin{theorem}
\label{thm:hard2}
There exists an absolute constant $\Delta \in \mathbb{N}$ such that Problem~\ref{prob:hyperdistance} is \APX-hard even when the input hypergraph $H=(V,\mathcal{E})$ is known to have maximum degree at most $\Delta$ and be such that $|e|\le 3$ for every $e\in \mathcal{E}$.



\end{theorem}

The proof of Theorem~\ref{thm:hard2} is given in Section~\ref{sec:hard2}.
Since the input hypergraph $H$ allows to compute the submodular function associated with $H$ efficiently, Theorem~\ref{thm:hard2}
implies the following.
\begin{corollary}
There exists a constant $\varepsilon >0$ such that Problem~\ref{pb:spgp} cannot be approximated to within a factor $(1+\varepsilon)$ in polynomial time, unless $\P=\NP$.    
\end{corollary}

On the positive side, we prove that shortest paths in the skeleton of $P_H$ can be approximated efficiently up to a multiplicative factor that is bounded by the \emph{maximum codegree} in~$H$. Given a hypergraph $H=(V,\mathcal{E})$ and a pair $\{u, v\}$ of distinct vertices, the \emph{codegree} of $\{u,v\}$ is defined as $$\text{codeg}(\{u,v\}):=|\{e \in \mathcal{E}|\{u,v\}\subseteq e\}|,$$ and the maximum codegree of $H$ is defined as $$\Delta_2(H):=\max\{\text{codeg}(\{u,v\})|u\neq v \in V\}.$$

\begin{theorem}\label{thm:approx}
There exists a polynomial-time algorithm that, given as input a hypergraph $H$ and two acyclic orientations $h_1, h_2$ of $H$, outputs a flip sequence $\{u_1,v_1\},\ldots,\{u_\ell,v_\ell\}$ from $h_1$ to $h_2$ such that $\ell \le \Delta_2(H)\cdot\text{dist}(h_1,h_2)$, where $\text{dist}(h_1,h_2)$ denotes the length of a shortest flip sequence from $h_1$ to $h_2$. 
\end{theorem}

In particular, for \emph{linear hypergraphs}, these are the hypergraphs $H=(V,\mathcal{E})$ satisfying $|e_1 \cap e_2|\le 1$ for every distinct $e_1, e_2 \in \mathcal{E}$, we obtain an efficient precise algorithm for computing a shortest flip sequence between two acylic orientations. 

\begin{corollary}
    There exists a polynomial-time algorithm that, given as input a linear hypergraph $H$ and two acyclic orientations $h_1, h_2$ of $H$, outputs a shortest flip-sequence from $h_1$ to $h_2$.
\end{corollary}

The proof of Theorem~\ref{thm:approx} is given in Section~\ref{sec:approx}.

\section{Proof of Theorem~\ref{thm:hard2}}
\label{sec:hard2}

Our proof describes a reduction from the minimum vertex cover problem in graphs, which is well-known to be \APX-hard. 
Recall that a vertex cover of a graph $G=(V,E)$ is a subset 
 $X\subseteq V$ such that every edge $e\in E$ has at least one endpoint in $X$.

We use the following classical hardness result due to H\r{a}stad~\cite{H01}. While H\r{a}stad only stated hardness of approximation on general graphs, his result in fact applies to bounded-degree graphs, see for instance~\cite{DS05}, page~449 for an explanation why we can restrict Theorem~\ref{thm:vertexcoverhardness} to the bounded-degree case. 

\begin{theorem}[\cite{H01}]\label{thm:vertexcoverhardness}
There exists an absolute constant $\Delta>0$ such that the following holds. Unless \P$=$\NP, there exists no polynomial-time algorithm that achieves an approximation ratio of less than $7/6$ for the following minimization problem.

\noindent\textbf{Input:} A graph $G=(V,E)$ of maximum degree at most $\Delta$.

\noindent\textbf{Output:} The size $|X|$ of a smallest vertex cover $X\subseteq V$ of $G$.
\end{theorem}

Before proceeding to the proof of Theorem~\ref{thm:hard2}, we describe a hypergraph that will be used as a gadget in our reduction. This gadget exploits the fact that stellohedra do not have the non-leaving face property~\cite{MP14,CP16}. The way this gadget is used in the reduction is inspired by previous reductions for the flip distance problem in geometric triangulations~\cite{AMP15,LP15}.

The \emph{star gadget} is the hypergraph $S=(V,\mathcal{E})$, where $V=\{0,1,\ldots ,6\}$ and $\mathcal{E}=\{ \{0, i\} : i\in [6]\} \cup \{\{0,i,j\} : i,j\in [6], i\not= j \}$. 
Hence it is composed of a star graph, together with all triples of vertices containing exactly two leaves of the star. The vertex 0 will be called the \emph{center} of the star gadget, and the vertices in $[6]$ will be referred to as the \emph{leaves}.

We define two orientations $h_1$ and $h_2$ of the star gadget. In both orientations, all hyperedges of the form $e=\{0,i\}$ are such that $h_1(e)=h_2(e)=i$. On the other hand, for hyperedges of the form $e=\{0,i,j\}$, we let $h_1(e)=\min \{i, j\}$, and $h_2(e)=\max \{i, j\}$. 

\begin{figure}
    \centering
    \includegraphics[page=5,width=\textwidth]{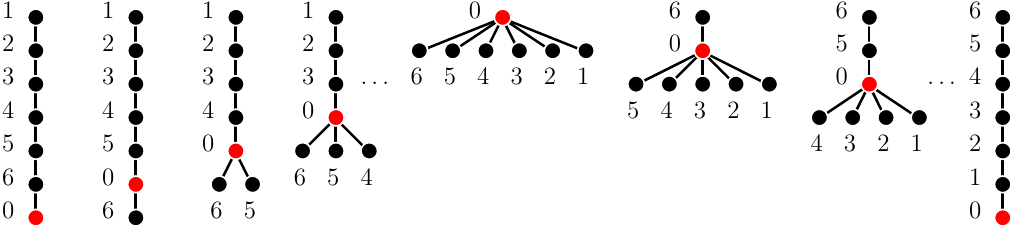}
    \caption{Twelve flips are sufficient to transform $h_1$ (left) into $h_2$ (right). The transitive reductions of the orientations are shown at each step, each of which is a broom, composed of a handle -- a chain above vertex 0 -- and a collection of independent vertices below 0.}
    \label{fig:star}
\end{figure}

\begin{lemma}
\label{lem:gadget}
    The flip distance between the two orientations $h_1$ and $h_2$ of the star gadget is 12.
    However, any sequence of flips between $h_1$ and $h_2$ that does not involve the vertex 0 has length at least 15.
\end{lemma}
\begin{proof}
    We observe that the acyclic orientations of the star gadget $S$ are one-to-one with \emph{brooms}, as pictured in Figure~\ref{fig:star}. Given an acyclic orientation $h$, let $W\subseteq [6]$ be the set of vertices $i$ whose edge $e=\{0,i\}$ is such that $h(e)=i$. Consider the transitive reduction $D$ of the digraph $D_h$. Then $D$ consists of two parts: A directed path, that we call the \emph{handle}, originating from 0 and involving vertices of $W$, and a directed star of center 0 with leaves consisting the vertices not in $W$, all pointing to 0. Indeed, the hyperedges of the form $e=\{0,i,j\}$ induce a total order on $W$, while the other vertices are incomparable in $D$.

    The orientations $h_1$ and $h_2$ are both such that the set $W$ contains all 6 vertices, but the order in the handle is reversed. There is a sequence of 12 flips that transforms $h_1$ into $h_2$. It consists of iteratively removing each of the 6 vertices from the handle of the broom, by performing the 6 flips $\{0,6\}, \{0,5\},\ldots ,\{0,1\}$, then the same flips again, in the same order: $\{0,6\}, \{0,5\},\ldots ,\{0,1\}$ (see Figure~\ref{fig:star}).

    On the other hand, if we are not allowed to flip any pair of the form $\{0,i\}$, we can only change the order of two adjacent vertices in the handle, by flipping a pair $\{i,j\}$ with $i,j>0$. This amounts to transform the identity permutation on 6 elements into the reverse identity permutation by adjacent transpositions. Since all pairs are reversed, this costs ${6\choose 2}=15$ flips.\qed
 \end{proof}

\begin{figure}
    \centering
    \includegraphics[page=1, scale=.8]{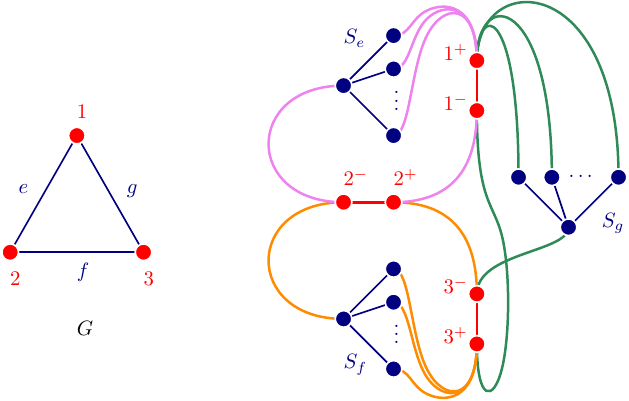}
    \caption{Illustration of the proof of Theorem~\ref{thm:hard2}. Given an instance $G$ of the vertex cover problem on the left, we construct the hypergraph $H$ on the right (we omit the depiction of the 3-edges within each of the star gadgets).}
    \label{fig:reduction}
\end{figure}

\begin{proof}[Theorem~\ref{thm:hard2}]
Consider an instance $G=(V,E)$ of the minimum vertex cover problem, with $|V|=n$ and $|E|=m$ and maximum degree bounded by a absolute constant $\Delta$.
We construct a hypergraph $H=(V',\mathcal{E})$ as follows.
First, we choose and fix some linear ordering $(V,\prec)$ of the vertex-set of $G$. Next, for each vertex $v\in V$, we construct two vertices $v^-$ and $v^+$ in $V'$, and connect them by an edge $\{v^-,v^+\}$ in $\mathcal{E}$.
Then, for each edge $e=\{u,v\}\in E$ such that $u\prec v$, we include the following vertices and hyperedges in the hypergraph $H$:
\begin{enumerate}
    \item A copy $S_e$ of the star gadget,
    \item \label{type:center} an edge connecting the center of the star $S_e$ to $u^-$,
    \item \label{type:middle} the edge $\{ u^+, v^-\}$,
    \item \label{type:leaves} 6 edges connecting $v^+$ to each of the leaves of $S_e$.
\end{enumerate}
The hypergraph $H$ has $|V'|=2n+7m$ vertices and $|\mathcal{E}|=n+29m$ hyperedges.
An illustration is given in Figure~\ref{fig:reduction}.

We now consider two orientations $h_1$ and $h_2$ of $H$. 
In both orientations, for each vertex $v\in V$, the edge $\{v^-,v^+\}$ is oriented from $v^-$ to $v^+$, and for each edge $e=\{u,v\}\in E$, the edges of $H$ of type~\ref{type:center} above are oriented from the center of $S_e$ to $u^-$, the edges of type~\ref{type:middle} are oriented from $u^+$ to $v^-$, and the edges of type~\ref{type:leaves} are oriented towards the leaves of $S_e$.
Also, for each star gadget, all hyperedges of size two are oriented towards the leafs.
The only difference between $h_1$ and $h_2$ are the orientations of the hyperedges of size three of the star gadgets. For each such hyperedge of the form $t=\{0,i,j\}$, we set $h_1(e)=\min \{i, j\}$, and $h_2(e)=\max \{i, j\}$.

\begin{claim}
    The orientations $h_1$ and $h_2$ are acyclic.
\end{claim} 

\begin{proof}
To certify that the digraphs $D_{h_1}$ and $D_{h_2}$ corresponding to the orientations $h_1, h_2$ of $H$ are indeed acyclic, we will describe two linear oderings of $V'$ that are \emph{topological orderings}, i.e., all directed edges in $D_{h_1}$ (or, $D_{h_2}$, respectively) have the property that their starting point precedes their endpoint in the ordering.

In both orderings, for $D_{h_1}$ and $D_{h_2}$, we start by listing all the $0$-vertices of all the star-gadgets $S_e$ for all $e \in E$. Note that these vertices form an independent set in both $D_{h_1}$ and $D_{h_2}$. Then, we start by listing the pairs of vertices $v^-,v^+$ (in this order) for all $v \in V$, where we order the pairs according to the linear order $\prec$ that we initially fixed on $V$. The latter makes sure that all the edges directed from $u^+$ to $v^-$ for some edge $\{u,v\} \in E$ with $u\prec v$ are indeed going forward in the linear ordering we define on $V'$. At the end of the ordering, we list all the vertices that are leafs of some star $S_e, e \in E$. While it does not matter in which order we put the leaves of different stars, it is crucial that within each of the stars, we order the leaves in decreasing order (i.e. $6,5,\ldots,1$) for the orientation $D_{h_1}$ and in the opposite order (i.e. $1,2,\ldots,6$) for the orientation $D_{h_2}$. It is not hard to check by going through all directed edges of $D_{h_1}$ and $D_{h_2}$ that both of these orderings for $D_{h_1}$ and $D_{h_2}$ are indeed topological, concluding the proof.\qed
\end{proof}

This finishes the description of our instance of the flip distance problem, composed of $H$ together with its two acyclic orientations $h_1$ and $h_2$.

Let us denote by $x$ the size of a minimum vertex cover of $G$. 
We now prove the following:

\begin{claim}
    The length of a shortest flip sequence between $h_1$ and $h_2$ is equal to
    \[
    2x + 12m.
    \]
\end{claim}    
\begin{proof}
    We first show that there exists a flip sequence transforming $h_1$ into $h_2$, of length at most $2x+12m$. First observe that all edges of $H$ of the form $\{v^-,v^+\}$ can be flipped independently, without affecting acyclicity (indeed, in the topological ordering of $V'$ for $D_{h_1}$ described in the proof of the previous claim, one may just flip the order in which $v^-$ and $v^+$ occur for every $v \in X$, and the arising ordering is still topological). Now pick a minimum vertex cover $X\subseteq V$ of $G$, and flip all edges of the form $\{v^+,v^-\}, v\in X$. This needs $x$ flips. 
    
    We now observe that after these flips, no directed path in the associated digraph leaves and goes back to a star gadget. 
    Indeed, consider any star gadget $S_e$ corresponding to an edge $e=uw$ of $G$ with $u\prec w$. Then $u \in X$ or $w \in X$. This means that in the orientation of $H$ obtained after flipping all pairs $\{v^+,v^-\}, v \in X$ we have that $u^-$ is a sink (if $u \in X$) or $w^+$ is a source (if $w \in X$). In the first case, every directed path leaving $S_e$ at its center must end already at $u^-$ and thus cannot return to $S_e$, and in the second case, every directed path entering $S_e$ at one of its leaves must begin already at $w^+$ and thhus could not have started in $S_e$.

    Hence all star gadgets can now be flipped independently from the rest of the hypergraph, without creating new directed cycles in the digraphs corresponding to the orientations.
    We next flip all star gadgets using exactly 12 flips per gadget, as described in Lemma~\ref{lem:gadget}. Finally, we can flip back all edges $\{v^-,v^+\}$ for $v\in X$. By the same arguments as above, these flips do not affect acyclicity.
    The total length of the described flip sequence is $2x+12m$.

    We now show that any flip sequence transforming $h_1$ into $h_2$ must have length at least $2x+12m$. Consider such a flip sequence of length $\ell$, and let $R$ be the set of edges $ab$ of size two of $H$, excluding the edges of the star gadgets, such that the pair $\{a,b\}$ occurs in the flip sequence from $h_1$ and $h_2$. Since the orientations $h_1$ and $h_2$ agree on these edges, and since they are not included in any larger hyperedges of $H$, the flip sequence must flip each such pair at least twice. So in total the flip sequence involves at least $2|R|$ such flips. We let $r=|R|$.

    We now partition the set $E$ of edges of $G$ into two subsets. For each edge $e=\{u,v\}\in E$ with $u\prec v$, consider the (six) paths from the center of $S_e$ to a leaf of $S_e$, and going through the vertices $u^-, u^+, v^-$, and $v^+$. 
    We now denote by $E_1\subseteq E$ the set of edges $e\in E$ such that no edge of any of these six paths is included in $R$, and let $E_2=E\setminus E_1$.
    
    From Lemma~\ref{lem:gadget}, each star gadget $S_e$ requires 15 flips if $e\in E_1$, since in that case any flip involving the center of the star would create a directed cycle with one of the paths associated with $e$, and thus the edges of size $2$ in $S_e$ can never be flipped in this case. On the other hand, if $e\in E_2$, we know that the star gadget requires at least 12 flips.
    
    Therefore the total length of the flip sequence satisfies
    \[
    \ell \geq 15|E_1| + 12|E_2| + 2r.
    \]
    We now observe that 
    \begin{equation}
    \label{eq:xbound}
    x\leq r + |E_1|.        
    \end{equation}
    Indeed, define a set $W\subseteq V$ of vertices as follows: We put $w \in W$ if and only if one of the edges that connect a center of a star gadget to $w^-$ is in $R$, or if $w^+$ is incident to an edge in $R$. It is straightfoward to check that for every edge $e \in E_2$, at least one of its endpoints is in $W$. Pause to note that $|W|\le |R|=r$. Thus, adding to $W$ for each edge $e\in E_1$ one of its endpoints, we obtain a vertex cover of $G$, which is of size at most $r+|E_1|$.
    
    It remains to bound $\ell$ as follows:
    \begin{eqnarray*}
        \ell & \geq & 15|E_1| + 12|E_2| + 2r \\
             & = & 15|E_1| + 12(m-|E_1|) + 2r \\
             & \ge & 2|E_1| + 12m + 2r \\
             & \geq & 2(x-r) + 12m + 2r \text{\ (from\ Eq. }\eqref{eq:xbound}\text{)} \\
             & = & 2x + 12m.
    \end{eqnarray*}    
This concludes the proof of the claim.\qed
\end{proof}

To finish the proof, we need to show that if there exists a $(1+\varepsilon)$-approximation of the flip distance between $h_1$ and $h_2$, then we can obtain a $(1+\delta)$-approximation of the size of a minimum vertex-cover of $G$, where $\delta=\delta(\varepsilon)$ tends to $0$ with $\varepsilon$. Let $d$ form a $(1+\varepsilon)$-approximation of the flip distance. Note that by the definition of a vertex-cover, it holds that $m=|E|\le \Delta x$. We have
\[
2x + 12m \leq d\leq (1+\varepsilon)(2x + 12m).
\]
Hence
\begin{eqnarray*}
    x & \leq & \frac 12 (d-12m) \\
    & \leq & \frac 12 ((1+\varepsilon)(2x + 12m) - 12m) \\
    & = & (1+\varepsilon) x + 6\varepsilon m \\
    & \leq & (1+\varepsilon + 6\varepsilon\Delta) x,
\end{eqnarray*}
so $\frac{d-12m}{2}$ forms a $(1+\delta)$-approximation for $x$ where $\delta(\varepsilon)=(1+6\Delta)\varepsilon=O(\varepsilon)$.
\qed\end{proof}

\section{Proof of Theorem~\ref{thm:approx}}
\label{sec:approx}

We now prove Theorem~\ref{thm:approx}. At the heart of our proof is the following key lemma, which we show first. 

\begin{lemma}\label{lemma:flippable}
    Let $H=(V,\mathcal{E})$ be a hypergraph and let $h_1, h_2$ be distinct acyclic orientations of $H$. Then there exists $e\in \mathcal{E}$ such that $h_1(e) \neq h_2(e)$ and $(h_1(e),h_2(e))$ is flippable in $h_1$. 
\end{lemma}
\begin{proof}
    Let $\mathcal{E}':=\{e \in \mathcal{E}|h_1(e)\neq h_2(e)\}$.
    Towards a contradiction, suppose that for every edge $e \in \mathcal{E}'$, the pair $(h_1(e),h_2(e))$ is not flippable in $h_1$. By definition, this means that for every $e \in \mathcal{E}'$ there exists a directed cycle in the digraph $D_{h^e}$, where $h^e:\mathcal{E}\rightarrow V$ denotes the orientation of $H$ defined by $$h^e(f):=\begin{cases}
      h_2(e), & \text{ if } h_1(f)= h_1(e) \text{ and } \{h_1(e),h_2(e)\} \subseteq f \\ h_1(f), & \text{ otherwise}  
    \end{cases}$$
    for every $f \in \mathcal{E}$.
    In the following, we denote by $D\subseteq D_{h_1}$ the \emph{transitive reduction} of $D_{h_1}$.
    Our first claim below is illustrated on Figure~\ref{fig:claim}.


\begin{figure}
    \centering
    \includegraphics[page=3, scale=.8]{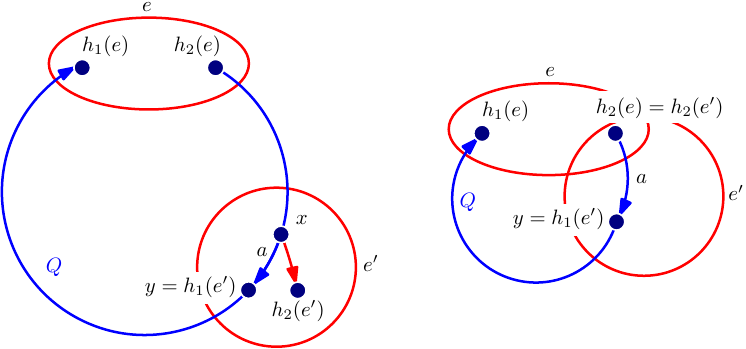}
    \caption{Illustration of the first claim in the proof of Lemma~\ref{lemma:flippable}. The figure on the right illustrates the case where $h_2(e)=h_2(e')$ and $a$ is the first arc of $Q$.}
    \label{fig:claim}
\end{figure}

\medskip

\begin{claim} 
For every $e \in \mathcal{E}'$, there exists some $e' \in \mathcal{E}'$ and a directed path $Q$ in $D$ starting at $h_2(e)$ and ending at $h_1(e)$, such that the following statements hold:
\begin{enumerate}
\item $Q$ has length at least two. 
\item There exists an arc $a=(x,y)$ on $Q$, such that $x,y \in e'$ and $h_1(e')=y$.   
\item The path $W$ composed of the prefix of $Q$ from $h_2(e)$ to $x$ followed by the arc $(x,h_2(e'))$ (if $x\not=h_2(e')$) is a directed path in $D_{h_2}$. It is of length at least one unless $h_2(e)=h_2(e')$ and $a$ is the first arc of $Q$.
\end{enumerate}
\end{claim}
\begin{proof}
Let $C$ be a shortest directed cycle in $D_{h^e}$. Then $C$ must use at least one arc of the form $(z,h_2(e))$
for some vertex $z\in f$ and $f\in \mathcal{E}$ with $\{h_1(e),h_2(e)\} \subseteq f$, $h_1(f)=h_1(e)$ and $h^e(f)=h_2(e)$.
Also, note that if $z=h_1(e)$, then $C$ must have length at least three, since $(h_2(e),h_1(e))$ cannot be an arc of $D_{h^e}$.
Observe that the directed path $P:=C-(z,h_2(e))$ in $D_{h^e}$ also forms a directed path from $h_2(e)$ to $z$ in the digraph $D_{h_1}$.
Now consider the path $P'$ in $D_{h_1}$, defined as $P':=P$ if $z=h_1(e)$, and as $P':=P+(z,h_1(e))$ if $z \neq h_1(e)$. Then $P'$ is a directed path in $D_{h_1}$ from $h_2(e)$ to $h_1(e)$. 
Since $C$ has length at least three if $z=h_1(e)$, $P'$ must have length at least two.

This proves that there is a directed path of length at least two in $D_{h_1}$ from $h_2(e)$ to $h_1(e)$. Let us now define $Q$ as a longest directed path in $D_{h_1}$ from $h_2(e)$ to $h_1(e)$. Then $Q$ also has length at least two, and the first item of the claim holds. Furthermore, the maximality of $Q$ implies that $Q$ is fully contained in the transitive reduction $D$ of $D_{h_1}$. 

Let $h_2(e)=q_0,q_1,\ldots,q_\ell=h_1(e)$ be the sequence of vertices of $Q$ in order. Then by definition of $D_{h_1}$, there exist edges $f_1,\ldots,f_\ell \in \mathcal{E}$ such that $q_{i-1}\in f_i$ and $h_1(f_i)=q_i$, for every $i=1,\ldots,\ell$. 
Suppose first that $h_1(f_i)=h_2(f_i)$ for every $i=1,\ldots,\ell$. Then $(q_{i-1},q_i)\in A(D_{h_2})$ for every $i=1,\ldots,\ell$, and clearly we also have $(h_1(e),h_2(e))=(q_\ell,q_0) \in A(D_{h_2})$. This is a contradiction to $h_2$ being an acyclic orientation of $H$. 
Consequently, at least one of $f_1, \ldots,f_\ell$ must be contained in $\mathcal{E}'$. Let $j$ be the smallest index such that $f_j \in \mathcal{E}'$. We now define $e':=f_j$ and $a=(x,y):=(q_{j-1},q_j)$. It follows directly from these definitions that the second item of our claim holds. 

Thus, only the third item remains to be verified, which states that 
$W$, defined as the prefix of $Q$ up to $x$, together with the arc $(x,h_2(e'))$ if $x\not= h_2(e')$, is a directed path in $D_{h_2}$.
To see this, note that by minimality of the choice of $j$, we have that $h_2(f_i)=h_1(f_i)=q_i$ and thus $(q_{i-1},q_i)\in A(D_{h_2})$ for every $0\le i\le j-1$. Furthermore, since $q_{j-1}\in e'=f_j$, we have $(q_{j-1},h_2(e'))=(q_{j-1},h_2(f_j)) \in A(D_{h_2})$ or $q_{j-1}=h_2(e')$. Thus the sequence $W$ 
is a directed walk from $h_2(e)=q_0$ to $h_2(e')$ in $D_{h_2}$. It is easy to see that $W$ is of positive length unless $h_2(e)=h_2(e')$ and $j=1$ (hence $a$ is the first arc on $Q$). This concludes the proof of the claim.
\qed\end{proof}

We next define an infinite sequence $(e_i)_{i=1}^\infty$ of hyperedges $e_i \in \mathcal{E}'$ of $H$, as follows. First we pick an element $e_1 \in \mathcal{E}'$ arbitrarily (note that $\mathcal{E}' \neq \emptyset$ as $h_1\neq h_2$). Then, for every $i \ge 1$, having defined $e_i$, we define $e_{i+1} \in \mathcal{E}'$ as the hyperedge $e'$ given by the previous claim, applied with $e_i$ in place of $e$. Concretely, we pick $e_{i+1}$ as an element of $\mathcal{E}'$ such that there exists
\begin{itemize}
\item a directed path $Q_i$ of length at least two in $D$ from $h_2(e_i)$ to $h_1(e_i)$,
\item an arc $a_i=(x_i,y_i)$ on $Q_i$ such that $x_i \in e_{i+1}$ and $h_1(e_{i+1})=y_i$,
\item and a directed walk $W_i$ in $D_{h_2}$ from $h_2(e_i)$ to $h_2(e_{i+1})$ which is of length at least $1$, unless $h_2(e_{i})=h_2(e_{i+1})$ and $a_i$ is the first arc of $Q_i$.
\end{itemize}

\begin{claim} 
There exists $i_0 \in \mathbb{N}$ such that for every $i \ge i_0$ we have that $h_2(e_i)=h_2(e_{i_0})$ and $a_i$ is the first arc of $Q_i$.
\end{claim}
\begin{proof}
Towards a contradiction, suppose there exist infinitely many distinct indices $i_1<i_2<i_3<\cdots$ such that for every $k \ge 1$, we either have $h_2(e_{i_k})\neq h_2(e_{i_k+1})$ or $a_{i_k}$ is not the first arc on $Q_{i_k}$. By the third item above, this means that $W_{i_k}$ is of length at least $1$ for every $k \ge 1$. Since $H$ is finite, there must exist distinct values $r<s$ such that $h_2(e_{i_{r}})=h_2(e_{i_{s}})$. But then the concatenation of the directed walks $W_{i_{r}},W_{i_{r}+1},\ldots,W_{i_s-1}$ forms a closed directed walk in $D_{h_2}$ which is of length at least $1$. This is a contradiction to the fact that $D_{h_2}$ is acyclic, and we thereby conclude the proof of our claim.\qed
\end{proof}

In the following, let us denote $v:=h_2(e_{i_0})$. By the previous claim, we have that $h_2(e_{i})=v$ and $a_i=(x_i,y_i)$ is the first arc on $Q_i$ for every $i \ge i_0$. In particular, $(h_2(e_{i_0+1}),h_1(e_{i_0+1}))=(v,y_{i_0})=(x_{i_0},y_{i_0})=a_{i_0}\in A(Q_{i_0})\subseteq A(D)$ is contained in $D$. However, we also have that $Q_{i_0+1}$ is a directed path in $D\subseteq D_{h_1}$ starting in $h_2(e_{i_0+1})$ and ending in $h_1(e_{i_0+1})$. Since $Q_{i_0+1}$ has length at least two by the first item above, it is a directed path also in $D_{h_1}-(h_2(e_{i_0+1}),h_1(e_{i_0+1}))$ parallel to the arc $(h_2(e_{i_0+1}),h_1(e_{i_0+1}))$. This is a contradiction, since by the above the arc $(h_2(e_{i_0+1}),h_1(e_{i_0+1}))$ is contained in the transitive reduction $D$ of $D_{h_2}$. This final contradiction shows that our initial assumption was incorrect, concluding the proof of the lemma by the principle of contradiction.
\qed\end{proof}

Using Lemma~\ref{lemma:flippable}, we are now ready to give the proof of Theorem~\ref{thm:approx}.
\begin{proof}[of Theorem~\ref{thm:approx}]
We describe an algorithm that takes as input a hypergraph $H=(V,\mathcal{E})$ and two distinct acyclic orientations $h_1, h_2$ of $H$, and returns a flip sequence $\{u_1,v_1\},\ldots,\{u_\ell,v_\ell\}$ transforming $h_1$ into $h_2$, as follows.

The algorithm starts by going through the list of hyperedges $e$ in $\mathcal{E}$ and checks whether $h_1(e)\neq h_2(e)$ and whether the pair $(h_1(e),h_2(e))$ is flippable in $h_1$. The latter can be done by checking whether the digraph $D_{h^e}$ defined as in the proof of Lemma~\ref{lemma:flippable} is acyclic (which can be done in polynomial time). Once the algorithm has found an edge $e 
 \in \mathcal{E}$ with $h_1(e) \neq h_2(e)$ such that $(h_1(e),h_2(e))$ is flippable in $h_1$ (such an edge $e$ exists by Lemma~\ref{lemma:flippable}), it proceeds to define $\{u_1,v_1\}:=\{h_1(e),h_2(e)\}$. If $h^e=h_2$, then the algorithm returns the flip sequence $\{u_1,v_1\}$ of length $1$ from $h_1$ to $h^e=h_2$. Otherwise, if $h^e\neq h_2$, it performs a recursive call on the input pair $h^e, h_2$ of distinct acyclic orientations, which gives back a flip sequence $\{u_1',v_1'\},\ldots,\{u_{\ell'}',v_{\ell'}'\}$ from $h^e$ to $h_2$. The algorithm now returns the concatenated flip sequence $\{u_1,v_1\},\{u_1',v_1'\},\ldots,\{u_{\ell'}',v_{\ell'}'\}$ from $h_1$ to~$h_2$. 

Note that with every recursive call of the above algorithm, the number of hyperedges $f$ for which the two considered orientations of $H$ assign different heads strictly decreases.
 
 Hence, the above algorithm returns in polynomial time a flip sequence of length at most $|\{e \in \mathcal{E}|h_1(e)\neq h_2(e)\}|$ from $h_1$ to $h_2$.
 
 We next show that
 $$|\{e \in \mathcal{E}|h_1(e)\neq h_2(e)\}| \le \Delta_2(H)\cdot \text{dist}(h_1, h_2),$$ where $\text{dist}(h_1, h_2)$ denotes the length of a shortest flip sequence transforming $h_1$ into $h_2$. To see this, note that if we flip at a pair $\{u,v\}$ of vertices in an orientation $h$ of $H$, by definition of a flip this can only change the value of $h(e)$ for those $e \in E$ with $\{u,v\}\subseteq e$. However, by definition of the codegree there are at most $\Delta_2(H)$ such hyperedges $e$ for every fixed pair $\{u,v\}$. Consequently, in an optimal flip sequence of length $\text{dist}(h_1,h_2)$ transforming $h_1$ to $h_2$, every single flip changes the heads of at most $\Delta_2(H)$ hyperedges, and in total at least $|\{e \in \mathcal{E}|h_1(e)\neq h_2(e)\}|$ head-assignments have to be changed. This shows that $|\{e \in \mathcal{E}|h_1(e)\neq h_2(e)\}| \le \Delta_2(H)\cdot\text{dist}(h_1,h_2),$ as desired. It follows that the above polynomial-time algorithm indeed returns a $\Delta_2(H)$-approximation of a shortest flip sequence from $h_1$ to $h_2$ in $H$. This concludes the proof.
\qed\end{proof}

\bibliographystyle{plain}
\bibliography{references}

\end{document}